\documentclass[reqno,a4paper,twocolumn,11pt,accepted=2024-01-11]{quantumarticle}
\pdfoutput=1
\usepackage{color}
\usepackage{hyperref}
\usepackage[latin1]{inputenc}
\usepackage[english]{babel}
\usepackage{amsmath,amssymb,amsfonts,amsthm,enumerate}
\usepackage{mathrsfs}
\usepackage[T1]{fontenc}
\usepackage{epsfig}
\usepackage{graphicx}
\usepackage[T1]{fontenc}
\usepackage{todonotes}

\definecolor{refkeybis}{gray}{.65}
\definecolor{labelkeybis}{gray}{.65}
{\makeatletter
\def\SK@refcolor{\color{refkeybis}}%
\def\SK@labelcolor{\color{labelkeybis}}}

\numberwithin{equation}{section} 
\newtheorem{theorem}{Theorem}[section]
\newtheorem{lemma}[theorem]{Lemma}

\newcommand{\N}{\mathbb{N}}
\newcommand{\R}{\mathbb{R}}
\newcommand{\C}{\mathbb{C}}

 \newcommand{\tauV}{{\kern-3pt\tau}}

 \newcommand{\oVVVk}{\overline{\mbox{\boldmath$V$}}\kern-3pt}
 \newcommand{\tVVVk}{\tilde{\mbox{\boldmath$V$}}\kern-3pt}

 \renewcommand{\SS}{\mathscr{S}}

 \usepackage{todonotes}
 \usepackage{chngcntr}
 \usepackage{authblk}
\counterwithout{equation}{section}
\begin{document}


 \title{Taming the Rotating Wave Approximation}

\author{Daniel Burgarth}

\affil{Department Physik, Friedrich-Alexander-Universit\"at Erlangen-N\"urnberg, Staudtstra\ss e 7, 91058 Erlangen, Germany}

\author{Paolo Facchi}

\affil{Dipartimento di Fisica, Universit\`{a} di Bari, I-70126 Bari, Italy, and INFN, Sezione di Bari, I-70126 Bari, Italy}

\author{Robin Hillier}

\affil{Department of Mathematics and Statistics, Lancaster University, Lancaster LA1 4YF, UK}
 \author{Marilena Ligab\`o}
\affil{Dipartimento di Matematica, Universit\`a di Bari, I-70125 Bari, Italy}

 \maketitle

\begin{abstract}
The interaction between light and matter is one of the oldest research areas of quantum mechanics, and a field that just keeps on delivering new insights and applications. With the arrival of cavity and circuit quantum electrodynamics we can now achieve strong light-matter couplings which form the basis of most implementations of quantum technology. But quantum information processing also has high demands requiring total error rates of fractions of percentage in order to be scalable (fault-tolerant) to useful applications. Since errors can also arise from modelling, this has brought into center stage one of the key approximations of quantum theory, the Rotating Wave Approximation (RWA) of the quantum Rabi model, leading to the Jaynes-Cummings Hamiltonian. While the RWA is often very good and incredibly useful to understand light-matter interactions, there is also growing experimental evidence of regimes where it is a bad approximation. Here, we ask and answer a harder question: for which experimental parameters is the RWA, although perhaps qualitatively adequate, already not good enough to match the demands of scalable quantum technology? For example, when is the error at least, and when at most, 1\%? To answer this, we develop rigorous non-perturbative bounds taming the RWA.

We find that these bounds not only depend, as expected, on the ratio of the coupling strength and the oscillator frequency, but also on the average number of photons in the initial state. This confirms recent experiments on photon-dressed Bloch-Siegert shifts. We argue that with experiments reporting controllable cavity states with hundreds of photons and with quantum error correcting codes exploring more and more of Fock space, this state-dependency of the RWA is increasingly relevant for the field of quantum computation, and our results pave the way towards a better understanding of those experiments.
\end{abstract}

The Rotating Wave Approximation (RWA) is one of the oldest and most important approximations in Quantum Theory. The starting point is at the birthplace of Nuclear Magnetic Resonance (NMR) in 1938, when Rabi and co-authors realized that rather than using rotating fields,  ``it is more convenient experimentally to use an oscillating field, in  which case the transition probability is approximately the same for weak oscillating fields near the resonance frequency''~\cite{rabi_new_1938}. This was significant: Rabi had shown earlier that the Schr\"odinger equation for rotating fields is easily solved analytically~\cite{rabi_space_1937}. This approximation was a crucial step in understanding driven quantum dynamics, 
as the time-dependent Schr\"odinger equation is notoriously hard to solve.

Perhaps this is the key reason for the popularity \cite{pop} of the RWA: it provides understanding and intuition of resonant driving. ~In fact, the importance of these ideas and the resulting techniques of NMR led to Rabi being awarded the Nobel Prize in Physics in 1944. But what justified the approximation, and how did Rabi get to it? Primarily reporting an experimental finding, Rabi himself does not provide justification, but over the last 80 years many different theoretical methods were used to provide justification and deeper understanding of the RWA (the literature is extensive, but see for instance~\cite{bloch_magnetic_1940,Shirley1965,haeberlen_coherent_1968,agarwal_rotating-wave_1973,burgarth_one_2022}).

Rabi described the atom as a two-level system and the field classically. In the full quantum description of light-matter interaction the situation is much more complicated. By the 1960s Quantum Electrodynamics was well established, and the electromagnetic field is now  itself a quantum system described by unbounded operators. Jaynes and Cummings~\cite{jaynes_comparison_1963} developed the full quantum mechanical version of the Rabi model (now called Quantum Rabi Model)
\begin{equation}
H=\frac{\Omega}{2} \sigma_z + \omega a^{\dagger} a+\lambda \sigma_x (a+a^{\dagger}),
\end{equation}
and applied the RWA to obtain the Jaynes-Cummings model 
\begin{equation}
H_{\textrm{RWA}}=\frac{\Omega}{2} \sigma_z +\omega a^{\dagger}a+\lambda (\sigma_+ a+ \sigma_- a^{\dagger}).
\end{equation}
Here, $\Omega$ is the energy difference between the two states of the atom, $\omega$ the light frequency and $\lambda$ the strength of the light-matter coupling; we always use $\hbar=1$.

\begin{figure}[t]
\centering
	\includegraphics[width=\columnwidth]{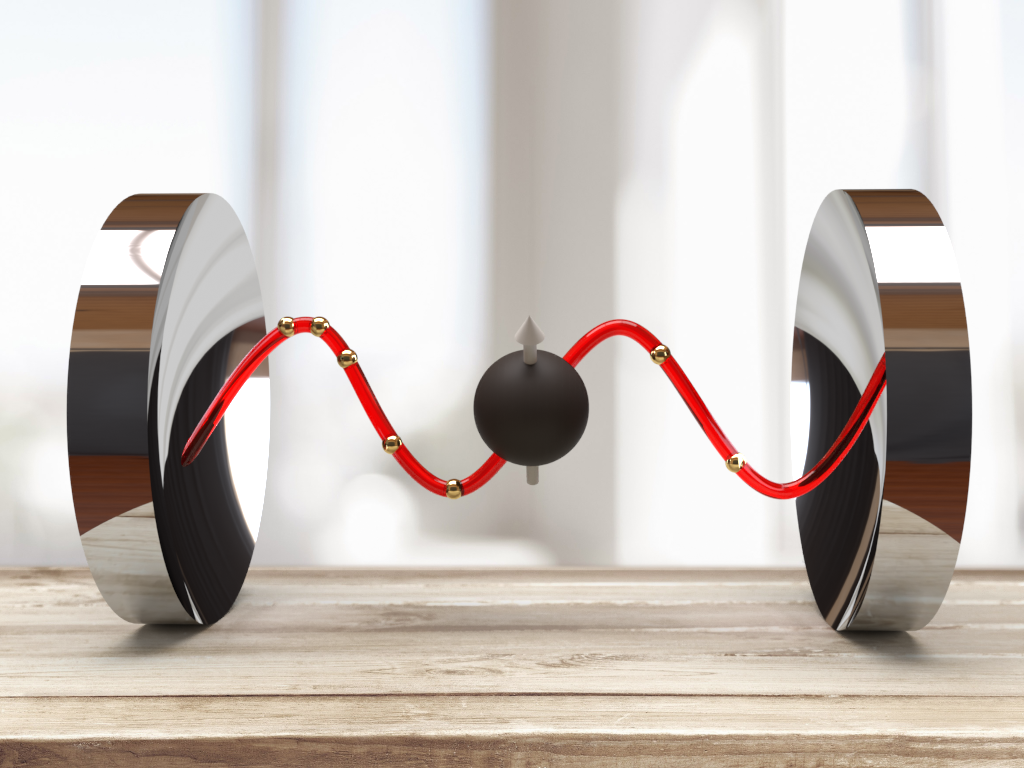}
	\caption{\label{nicepic}Light-matter interactions have been a major driver in quantum physics for half a decade. Often, atoms are placed into cavities to amplify their effective coupling strength with photons. Here, we show that the rotating wave approximation is not only determined by such coupling strength and the frequency of the driving, but also by the number of photons (naively depicted as golden spheres) in the cavity.}
\end{figure}
	
Due to its simplicity and wide range of applicability, the Jaynes-Cummings model is the main work horse of light-matter interactions and, by extension, quantum technology. For an excellent overview of its scope see~\cite{larson_jaynes-cummings_2021}. While at the time of the original paper the RWA was rather natural, given that the bare coupling between matter and light tends to be extremely weak, in cavity and circuit QED nowadays it is well understood that the effective coupling can be enhanced to a level where the RWA breaks down. This is often referred to as the Ultrastrong Coupling regime. For examples of experiments, see~\cite{forn-diaz_observation_2010} and~\cite{li_vacuum_2018}, for a recent review see \cite{Nori2019}.

While there is no rigorous derivation of the RWA for the Jaynes-Cummings model till date, the common lore is that the ratio $g\equiv \lambda /\omega$ between the light-matter coupling and the light frequency is the key parameter \cite{larson_jaynes-cummings_2021}. This is motivated by perturbative arguments and of course backed up by extensive numerical studies and simulations. For a summary of the different regimes see Table 1.1 in~\cite{larson_jaynes-cummings_2021}, where it is argued that for $g\approx 0.1$ the RWA breaks down.
\begin{figure}[t]
\centering
	\includegraphics[width=\columnwidth]{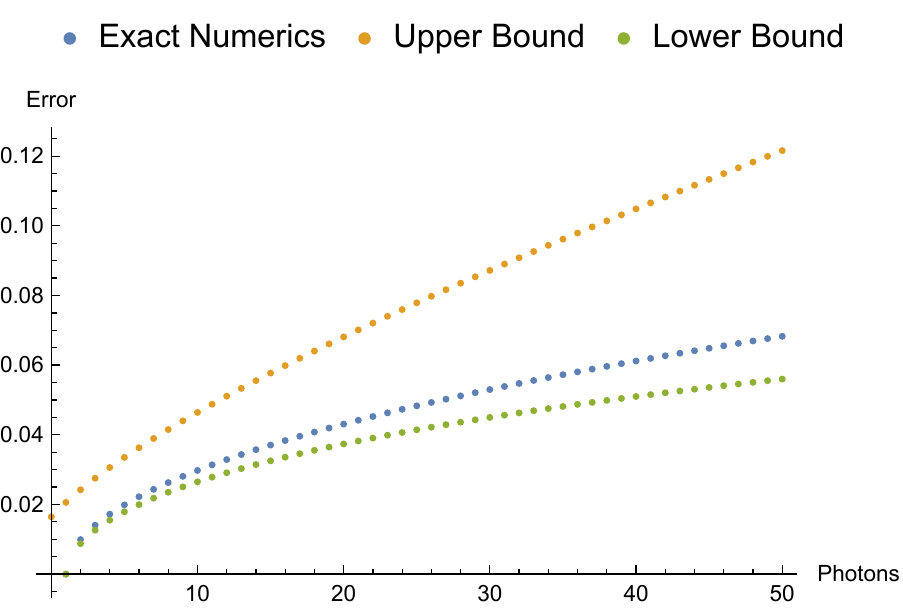}
	\caption{\label{fig:upperlower}Bounding the error of the RWA. We consider a Fock state evolving under the quantum Rabi and the Jaynes-Cummings model, respectively. We show our analytical upper and lower bounds and the exact numerical norm difference between the two models. We see that the error grows with the photon number, and that the bounds provide a good understanding of the scaling  (other parameters here $g=\frac{\lambda}{\omega}=\frac{1}{100}$, $\Delta =0$, $t\approx 0.04/\omega$).}
\end{figure}
On the other hand, this picture changes for high photon numbers. Indeed, Walls showed~\cite{walls_higher_1972} that the Bloch-Siegert shift (taken as a sign of the breakdown of the RWA) scales with the number of photons. This was also observed experimentally~\cite{wang_photon-dressed_2020}.  See also~\cite{puri_mathematical_2011} for a perturbative argument that $\lambda\sqrt{\langle a^{\dagger} a\rangle }\ll \omega$ is a more relevant condition in that regime.

What this means is that the quality of the RWA does not only depend on the parameters of the model, but also on the initial state of the system. See Fig.~\ref{fig:upperlower} for a numerical example. Indeed, we prove that there are short times  $t\le \pi/\omega$ for which
 \begin{equation}\label{lb}
 \|e^{-itH}-e^{-itH_{\textrm{RWA}}}\|\ge \frac{1}{6}
 \end{equation}
  for $any$ parameter value \cite{remark}. This should be considered as a big error, because the biggest difference between two unitaries is $2$ and because modern quantum technology demands errors well below $1\%$ (see below). Does this mean that the RWA is wrong? No, because we also show that for any state $\varphi$ and any time $t$,
\begin{equation}\label{eq:main}
e^{-itH}\varphi-e^{-itH_{\textrm{RWA}}}\varphi\to 0, \quad \text{as } g\to 0.
\end{equation}
This is our main result, providing a rigorous justification to the RWA. It does not contradict Eq.~(\ref{lb}), but is a typical phenomenon of unbounded Hamiltonians such as $H$ and $H_{\textrm{RWA}}$: there is no norm convergence, only state-dependent convergence. This is one the key technicalities that make it hard to apply standard perturbative arguments for the RWA.

Let us discuss the relevance of this photon-dependence in the context of quantum technology.  For fault-tolerant quantum computation, very high fidelity with error rates  $<10^{-3}$ are required  \cite{Gottesman2022}.  Moreover, modern qubit designs such as GKP \cite{Grimsmo2021} and CAT qubits use cavity states and explore high numbers of photons. In particular, CAT states have been created with about 100 photons~\cite{vlastakis_deterministically_2013}.
It is therefore necessary to have a good handle of the error of the RWA. Since quantum algorithms also invoke dynamics, it is not sufficient to simply match spectral properties, as it is usually done, but we need to bound the difference in evolution operators. The interesting evolution time regime here are short times up to $\pi/\omega$: already there, the RWA dynamics can deviate substantially. We show that the maximal error $\epsilon_n$ that the RWA has for an $n-$photon Fock state in a short time interval up to $\pi /\omega $ is bounded between 
\begin{equation}
\label{lowerbound2}
5g\sqrt{n +3} \ge \epsilon_n  \ge \frac{1}{6}-\frac{1}{216 g^2 n}-\frac{7}{12 n},
\end{equation}
proving that the RWA becomes good for small $g$ but bad for large $n$. Tighter and more general bounds and the full proofs of our results are provided in the Appendix. See also Fig.~\ref{fig:upperlower} for numerical examples of these refined bounds. These bounds prove that $g\sqrt{n}$ is the right parameter (as anticipated by the perturbative argument~\cite{puri_mathematical_2011}) for the validity of the RWA for Fock states. For more general states, see the Appendix. These bounds will be useful for experimentalists in quantum information to judge if they should apply the RWA or not.

We would now like to explain the idea which allows us to tame the RWA. Although there are many different conceptual ideas trying to justify the RWA, almost all of them agree that `highly oscillatory terms' in a Hamiltonian may sometimes be discarded to a good approximation. But why? Interestingly, some have argued that such terms are not observable, since measurements take finite time. This is plausible; however it turns out that even if measurements are instantaneous, the RWA can be taken. Others argue on the basis of first order perturbation theory, when the term involves an integral over the Hamiltonian. This gives a good qualitative picture but makes it impossible to compute a rigorous and precise picture. In a more recent work~\cite{burgarth_one_2022} a different route was taken: by an integration by part, the difference between two evolutions can indeed be written in terms of an integral over the difference of their generating Hamiltonians, where fast oscillations average out. This allows one to prove and provide bounds for the RWA, but only in the finite dimensional case. Here, we develop an integration by part to unbounded operators. In the general case, this is hard, so we are employing several structures of the specific problem of the quantum Rabi model to simplify the analysis. First, both $H$ and $H_{\textrm{RWA}}$ are time-independent, so we can use the rich theory of semigroups. Secondly, $H_{\textrm{RWA}}$ has many conserved quantities and can only increase and decrease the photon number by one. Finally, all involved quantities are well-defined on the subspace of rapidly decreasing functions and leave it invariant, which allows us to work on that subspace. We refer to the Appendix for the mathematical details.

To summarize, after decades of work and conjectures around the RWA for the highly relevant quantum Rabi model, we have now got a rigorous proof and in addition a complete quantitative measure in terms of lower and upper bounds on the error of the approximation. In particular, this confirms the experimental and numerical findings that the error becomes large for large ratio $g$ between light-matter coupling and light frequency or for large photon numbers and hence the dependence on the state of the system. In practice, for given fixed photon number and given maximally permissible error this tells us how small $g$ has to be in order for the RWA to work. Since experiments are working with ever growing systems, our results will be of immediate relevance to the understanding and setup of those experiments and further developments in quantum technology. We expect that the methods developed for our proof can be applied to tame the RWA for other interesting models, such as systems with multiple modes, nonlinearities and other descendants of the Jaynes-Cummings model \cite{larson_jaynes-cummings_2021}. Finally, we remark that the RWA is also important in quantum control, where it can be used to understand selective population transfer via frequency tuning \cite{dda,chambrion,augier,robin}. 
It will be interesting to apply our bounds to this scenario in future studies.

\acknowledgments
DB acknowledges funding by the Australian Research Council (project numbers FT190100106, DP210101367, CE170100009) and the Munich Quantum Valley. PF and ML were partially supported by the Italian National Group of Mathematical Physics (GNFM-INdAM), by Istituto Nazionale di Fisica Nucleare (INFN) through the project ``QUANTUM'', and by Regione Puglia and QuantERA ERA-NET Cofund in Quantum Technologies (GA No. 731473), project PACE-IN. ML acknowledges the support by PNRR MUR project PE0000023- NQSTI.

\onecolumn

\part*{Appendix}
\section{Time evolution of the Rabi and the Jaynes-Cummings models}
We consider the infinite dimensional Hilbert space $L^2(\R)$, the creation operator $a^{\dagger}=\frac{1}{\sqrt{2}}\left( x- \frac{d}{dx}\right)$ and the annihilation operator $a=\frac{1}{\sqrt{2}}\left( x+ \frac{d}{dx}\right)$ on Schwartz space $\SS(\R)$. A fundamental feature of these two operators is that their commutator is the identity operator, i.e.
\begin{equation}
[a,a^{\dagger}]=I.
\end{equation}
Now we consider the following two Hamiltonians
\begin{equation}
H=\frac{\Omega}{2} \sigma_z \otimes I+ I \otimes \omega a^{\dagger}a+\lambda \sigma_x \otimes (a+a^{\dagger})
\end{equation}
and
\begin{equation}
H_{\textrm{RWA}}=\frac{\Omega}{2} \sigma_z \otimes I+ I \otimes \omega a^{\dagger}a+\lambda (\sigma_+ \otimes a+ \sigma_-\otimes a^{\dagger})
\end{equation}
on $\C^2\otimes \SS(\R)$, and we also denote their closures by $H$ and $H_{\textrm{RWA}}$, with suitable dense domains. Here $\lambda, \omega, \Omega \in \R$,
\begin{equation}
\sigma_x=\left( \begin{array}{cc} 0 & 1 \\ 1 & 0 \end{array}\right), \quad \sigma_y=\left( \begin{array}{cc} 0 & -i \\ i & 0 \end{array}\right), \quad \sigma_z=\left( \begin{array}{cc} 1 & 0 \\ 0 &-1 \end{array}\right),
\end{equation}
are the Pauli matrices and
\begin{equation}
\sigma_+=\left( \begin{array}{cc} 0 & 1 \\ 0 & 0 \end{array}\right), \quad \sigma_-=\left( \begin{array}{cc} 0 & 0 \\ 1 & 0 \end{array}\right).
\end{equation}
In what follows, we usually work on $\C^2\otimes \SS(\R)$ without saying so explicitly every time, as this subspace of $\C^2\otimes L^2(\R)$ forms a common invariant core of the operators we are studying in this section; this can be shown following the standard methods in~\cite[Sec.X.6]{RS2}. To simplify the notation, we usually use the same notation for an operator on this core and its closure on the whole domain.

\subsection{Interaction picture: time dependent Hamiltonians $H_1(t)$ and $H_2(t)$}
We consider the following Hamiltonian 
\begin{equation}
H_0=\frac{\omega}{2} \sigma_z \otimes I+ I \otimes \omega a^{\dagger}a,
\end{equation}
and define for all $t \in \R$
\begin{equation}
U_1(t)=e^{it H_0}e^{-it H}, \quad U_2(t)=e^{it H_0}e^{-it H_{\textrm{RWA}}}.
\end{equation}
We have that for all $j \in \{1,2\}$: $U_j(0)=I$ and
\begin{equation}\label{eqn:SEtime}
i  \frac{d U_j(t)}{dt}=H_j(t) U_j(t),
\end{equation}
where for all $t\in \R$:
\begin{equation}
H_1(t)=e^{it H_0}(H-H_0)e^{-it H_0}, \quad H_2(t)=e^{it H_0}(H_{\textrm{RWA}}-H_0)e^{-it H_0},
\end{equation}
We define the detuning between field and atom as $\Delta=\Omega-\omega$ and compute $H_1(t)$:
\begin{eqnarray}
H_1(t)&=& e^{it H_0}(H-H_0)e^{-it H_0} \nonumber \\
         &=& e^{it H_0}\left(\frac{\Delta}{2} \sigma_z \otimes I + \lambda \sigma_x \otimes (a+a^{\dagger})\right)e^{-it H_0} \nonumber \\
         & = &  \frac{\Delta}{2} \sigma_z \otimes I+ \lambda (e^{it \omega \sigma_z /2} \sigma_x e^{-it \omega \sigma_z /2}) \otimes (e^{it \omega a^{\dagger} a}(a+a^{\dagger})e^{-it \omega a^{\dagger} a})
\end{eqnarray}
We get
\begin{eqnarray}
e^{it \omega \sigma_z /2} \sigma_x e^{-it \omega \sigma_z /2}= \cos \left(t \omega \right) \sigma_x- \sin\left(t \omega \right) \sigma_y,
\end{eqnarray}
and
\begin{equation}
	e^{it \omega a^{\dagger} a}ae^{-it \omega a^{\dagger} a}=e^{-it\omega } a, \qquad e^{it \omega a^{\dagger} a}a^{\dagger}e^{-it \omega a^{\dagger} a}=e^{it\omega} a^{\dagger}.
\end{equation}
Therefore,
\begin{eqnarray}
H_1(t)&=&  \frac{\Delta}{2} \sigma_z \otimes I+ \lambda \left(\cos \left(t \omega \right) \sigma_x- \sin\left(t \omega \right) \sigma_y \right) \otimes \left( e^{-it\omega } a+e^{it\omega} a^{\dagger}\right) \nonumber \\
          &=& \frac{\Delta}{2} \sigma_z \otimes I+ \lambda \left(\sigma_+ \otimes a + \sigma_- \otimes a^{\dagger}+ e^{2it\omega}\sigma_+ \otimes a^{\dagger} + e^{-2it\omega}\sigma_- \otimes a\right). \nonumber \\
\end{eqnarray}
In a similar way we compute $H_2(t)$:
\begin{eqnarray}
H_2(t)&=& e^{it H_0}(H_{\textrm{RWA}}-H_0)e^{-it H_0} \nonumber \\
         &=& e^{it H_0}\left(\frac{\Delta}{2} \sigma_z \otimes I + \lambda( \sigma_+ \otimes a+\sigma_- \otimes a^{\dagger})\right)e^{-it H_0} \nonumber \\
        & = &  \frac{\Delta}{2} \sigma_z \otimes I + \lambda  (\sigma_+\otimes  a+  \sigma_-  \otimes   a^{\dagger} ).
\end{eqnarray}
We notice that $H_2$ is time-independent, and again we use the same symbol for its closure.

\subsection{Computation of  $U_2(t)$}

Notice that, since $H_2$ is time-independent, $\{U_2(t)\}_{t \in \R}$ is a unitary group: for all $t \in \R$,
\begin{equation}
U_2(t)=e^{itH_0}e^{-itH_{\textrm{RWA}}}=e^{-it \left(\frac{\Delta}{2} \sigma_z \otimes I + \lambda  (\sigma_+\otimes  a+  \sigma_-  \otimes   a^{\dagger} )\right)}= e^{-it H_2}.
\end{equation} 
We observe that $\C^2 \otimes \SS(\R) \subset D(H_2)$ is a set of analytic vectors for $H_2$. Let 
\begin{equation}
e_1=\left( \begin{array}{c} 1\\ 0  \end{array}\right), \qquad e_2=\left( \begin{array}{c} 0\\ 1  \end{array}\right)
\end{equation}
be the canonical orthonormal basis of $\C^2$. Using the following identities  
\begin{equation}
\sigma_+e_1=\left( \begin{array}{c} 0\\ 0  \end{array}\right), \quad \sigma_+ e_2=e_1, \quad  \sigma_-e_1=e_2, \quad \sigma_- e_2=\left( \begin{array}{c} 0\\ 0  \end{array}\right),
\end{equation}
we compute the even and the odd powers of $H_2$ on vectors $e_1 \otimes \psi$ and $e_2 \otimes \psi$, with $\psi \in \SS(\R)$. For all $j \in \N$,
\begin{equation}
H_2^{2j}(e_1 \otimes \psi)= \sum_{\ell=0}^j \left( \begin{array}{c}j \\ \ell \end{array}\right) \lambda^{2\ell} \left( \frac{\Delta}{2}\right)^{2(j-\ell)}e_1 \otimes (a a^{\dagger})^\ell \psi,
\end{equation}
\begin{equation}
H_2^{2j}(e_2 \otimes \psi)= \sum_{\ell=0}^j \left( \begin{array}{c}j \\ \ell \end{array}\right) \lambda^{2\ell} \left( \frac{\Delta}{2}\right)^{2(j-\ell)}e_2 \otimes (a^{\dagger} a)^\ell \psi,
\end{equation}
\begin{equation}
H_2^{2j+1}(e_1 \otimes \psi)= \sum_{\ell=0}^j \left( \begin{array}{c}j \\ \ell \end{array}\right) \left[ \lambda^{2\ell} \left( \frac{\Delta}{2}\right)^{2(j-\ell)+1}e_1 \otimes (a a^{\dagger})^\ell \psi  +\lambda^{2\ell+1} \left( \frac{\Delta}{2}\right)^{2(j-\ell)}e_2 \otimes a^{\dagger} (a a^{\dagger})^\ell \psi \right],
\end{equation}
\begin{equation}
H_2^{2j+1}(e_2 \otimes \psi)= \sum_{\ell=0}^j \left( \begin{array}{c}j \\ \ell \end{array}\right) \left[ -\lambda^{2\ell} \left( \frac{\Delta}{2}\right)^{2(j-\ell)+1}e_2 \otimes ( a^{\dagger} a)^\ell \psi  +\lambda^{2\ell+1} \left( \frac{\Delta}{2}\right)^{2(j-\ell)}e_1 \otimes a (a^{\dagger}a )^\ell \psi \right].
\end{equation}

\begin{lemma}\label{lemma:U2}
$U_2(t)(\C^2 \otimes \SS(\R)) \subset \C^2 \otimes \SS(\R)$ for all $t \in \R$.
\end{lemma}
\begin{proof}
By the $N$-representation theorem for $\SS(\R)$~\cite[Thm.V.13]{RS1}, we have that $\psi \in \SS(\R)$ if and only if for all $m \in \N$:
\begin{equation}
\sup_{n \in \N} |\langle \varphi_n | \psi \rangle |n^m <+\infty,
\end{equation}
where $\{\varphi_n\}_{n \in \N}$ is the orthonormal eigenbasis of the number operator $a^{\dagger}a$, i.e. $a^{\dagger}a \varphi_n= n \varphi_n$ for all $n \in \N$. 
We have that for all $n \in \N$ and $t \in \R$:
\begin{equation}
U_2(t)(e_1 \otimes \varphi_n)=e_1 \otimes a_n(t) \varphi_n + e_2 \otimes b_{n}(t) \varphi_{n+1}
\end{equation}
and 
\begin{equation}
U_2(t)(e_2 \otimes \varphi_n)=e_1 \otimes c_n(t) \varphi_{n-1} + e_2 \otimes d_{n}(t) \varphi_{n}
\end{equation}
where
\begin{equation}
a_n(t)=\cos\left( t \sqrt{\lambda^2(n+1)+ \left( \frac{\Delta}{2}\right)^2}\right)- \frac{i\Delta}{2}\frac{\sin \left( t \sqrt{\lambda^2(n+1)+ \left( \frac{\Delta}{2}\right)^2}\right) }{ \sqrt{\lambda^2(n+1)+ \left( \frac{\Delta}{2}\right)^2}},
\end{equation}
\begin{equation}
b_n(t)=-i \lambda \sqrt{n+1}\frac{\sin \left( t \sqrt{\lambda^2(n+2)+ \left( \frac{\Delta}{2}\right)^2}\right)}{ \sqrt{\lambda^2(n+2)+ \left( \frac{\Delta}{2}\right)^2}},
\end{equation}
\begin{equation}
c_n(t)=-i \lambda \sqrt{n}\frac{\sin \left( t \sqrt{\lambda^2n+ \left( \frac{\Delta}{2}\right)^2}\right)}{ \sqrt{\lambda^2n+ \left( \frac{\Delta}{2}\right)^2}}
\end{equation}
and
\begin{equation}
d_n(t)=\cos\left( t \sqrt{\lambda^2n+ \left( \frac{\Delta}{2}\right)^2}\right)+ \frac{i\Delta}{2}\frac{\sin \left( t \sqrt{\lambda^2n+ \left( \frac{\Delta}{2}\right)^2}\right) }{ \sqrt{\lambda^2n+ \left( \frac{\Delta}{2}\right)^2}}.
\end{equation}
Let $\psi \in \SS(\R)$ and  $t \in \R$, then
\begin{equation}
U_2(t)(e_1 \otimes \psi)=e_1 \otimes \psi_1 + e_2 \otimes \psi_2, \quad U_2(t)(e_2 \otimes \psi)=e_1 \otimes \psi_3 + e_2 \otimes \psi_4,
\end{equation}
where 
\begin{equation}
\psi_1= \sum_{n =0}^{+\infty} \langle \varphi_n | \psi \rangle a_n(t) \varphi_n, \quad \psi_2=\sum_{n =0}^{+\infty} \langle \varphi_n | \psi \rangle b_n(t) \varphi_{n+1},
\end{equation}
and
\begin{equation}
\psi_3= \sum_{n =1}^{+\infty} \langle \varphi_n | \psi \rangle c_n(t) \varphi_{n-1}, \quad \psi_4=\sum_{n =0}^{+\infty} \langle \varphi_n | \psi \rangle d_n(t) \varphi_{n}.
\end{equation}

Notice that for all $m \in \N$ and for all $j\in\{1,2,3,4\}$:
\begin{equation}
\sup_{n \in \N} |\langle \varphi _n|\psi_j\rangle | n^m <+\infty,
\end{equation}
hence $\psi_j \in \SS(\R)$ and therefore $U_2(t)(\C \otimes \SS(\R)) \subset \C \otimes \SS(\R)$.
\end{proof}

\begin{lemma}\label{lemmaaction}
For all $t \in \R$ and $\Psi \in \C^2 \otimes \SS(\R) $, we have
\begin{equation}
\int_0^t (H_2-H_1(s) )\Psi \,ds=-\frac{\lambda \sin \left(t \omega \right) }{\omega} \left(e^{it\omega}\sigma_+ \otimes a^{\dagger} + e^{-it\omega}\sigma_- \otimes a \right)\Psi,
\end{equation}
and we denote the closure of this operator by $S_{21}(t)$. Moreover:
\begin{itemize}
\item $S_{21}(t)(\C^2 \otimes \SS(\R))  \subset \C^2 \otimes \SS(\R)$ for all $t \in \R$;
\item for all $\Psi \in \C^2 \otimes \SS(\R) $ and for all $t \in \R$: $\frac{d}{dt}S_{21}(t)\Psi = (H_2(t)-H_{1}(t))\Psi$.
\end{itemize}
\end{lemma}
\begin{proof}
On $\C^2 \otimes \SS(\R)$, we have
\begin{eqnarray} 
S_{21}(t)&:=&\int_0^t (H_2-H_1(s) )\,ds \nonumber \\
              &=&- \lambda \int_0^t \left(e^{2is\omega}\sigma_+ \otimes a^{\dagger} + e^{-2is\omega}\sigma_- \otimes a \right)\,ds \nonumber \\
              &=& -\frac{\lambda}{2i \omega }\left[e^{2is\omega}\right]_{s=0}^{s=t} \sigma_+ \otimes a^{\dagger}+\frac{\lambda }{2i \omega}\left[e^{-2is\omega}\right]_{s=0}^{s=t}\sigma_- \otimes a \nonumber \\
              &=&-\frac{\lambda }{ 2i\omega}\left(e^{2it\omega}-1\right)\sigma_+ \otimes a^{\dagger}+\frac{\lambda}{2i\omega}\left(e^{-2it\omega}-1\right)\sigma_- \otimes a \nonumber \\
              &=&-\frac{\lambda \sin \left(t \omega \right) }{\omega} \left(e^{it\omega}\sigma_+ \otimes a^{\dagger} + e^{-it\omega}\sigma_- \otimes a \right). 
\end{eqnarray} 
\end{proof}

\begin{lemma}\label{eqn:intparts1}
For all $\Psi \in \C^2 \otimes \SS(\R)$:
\begin{eqnarray}
i (U_2(t)-U_1(t))\Psi&=&  S_{21}(t)U_2(t)\Psi + \\
                                  &&+  i\int_{0}^t U_1(t)U_1(s)^\dagger (S_{21}(s)H_2-H_1(s)S_{21}(s))U_2(s)\Psi \, ds, \nonumber 
\end{eqnarray}
\end{lemma}
\begin{proof}
Let $\Psi \in \C^2 \otimes \SS(\R)$, by Lemmas~\ref{lemma:U2} and~\ref{lemmaaction} we have that for all $s \in \R$:
\begin{equation}
U_{2}(s)\Psi, S_{21}(s)H_2U_{2}(s)\Psi, H_1(s)S_{21}(s)U_2(s)\Psi \in \C^2\otimes \SS(\R).
\end{equation}
By equation~(\ref{eqn:SEtime}) we have that
\begin{eqnarray}
i (U_2(t)-U_1(t))\Psi &=& iU_1(t)(U_1(t)^\dagger U_2(t)-I)\Psi \nonumber \\
                                                  &=& iU_1(t)\left[ U_1(s)^{\dagger}U_2(s)\right]_0^t\Psi \nonumber \\
                                                  &=& iU_1(t)\int_{0}^t \frac{d}{d s}U_1(s)^{\dagger}U_2(s)\Psi  \, ds \nonumber \\
                                                 &=& U_1(t)\int_{0}^t U_1(s)^{\dagger}(H_2-H_1(s))U_2(s)\Psi \, ds. 
\end{eqnarray}
We observe that for all $s \in \R$:
\begin{eqnarray}
\frac{d}{ds}\left( U_1(s)^{\dagger} S_{21}(s)U_{2}(s)\Psi \right)&=& iU_1(s)^{\dagger}(H_1(s)S_{21}(s)-S_{21}(s)H_2)U_2(s)\Psi + \nonumber\\
 && +U_1(s)^{\dagger}(H_2-H_1(s))U_2(s)\Psi ,
\end{eqnarray}
therefore 
\begin{eqnarray}
i (U_2(t)-U_1(t))\Psi &=& U_1(t)\int_{0}^t U_1(s)^{\dagger}(H_2-H_1(s))U_2(s)\Psi \, ds \nonumber \\
                                 &=& S_{21}(t)U_2(t)\Psi \nonumber \\
                                 && +i  \int_{0}^t U_1(t)U_1(s)^{\dagger}(S_{21}(s)H_2-H_1(s)S_{21}(s))U_2(s)\Psi \, ds .
\end{eqnarray}

\end{proof}

Notice that a similar Lemma might hold with $U_1(t)$ and $U_2(t)$ interchanged. This would however be much harder to prove, as our current prove relies on the simple structure of the Jaynes-Cummings interaction through Lemma~\ref{lemma:U2}.
\section{Computation of bounds and the rotating wave approximation}
Without loss of generality, we can assume  $\lambda, \Omega, \omega > 0$.
\subsection{Upper bound for generic vectors}
\begin{theorem}
For all $\Psi \in \C^2 \otimes \SS(\R)$ and $t \in \R$:
\begin{equation}\label{upperbound}
\left\Vert (U_{2}(t)-U_{1}(t))\Psi \right\Vert  \le 	 \frac{\lambda}{\omega} \biggl[
	\| (N +2)^{1/2} \Psi \| 
	+ |t| \Bigl( |\Delta| \|(N+2)^{1/2} \Psi \| + 3 \lambda \|\bigl((N+2)(N+3)\bigr)^{1/2}\Psi \|\Bigr)	\biggr],
\end{equation}
where $N=I \otimes a^{\dagger}a$. Moreover for all $\Psi \in \C^2 \otimes L^2(\R)$:
\begin{equation}\label{approxim}
\lim_{\omega \to + \infty} \left\|  (e^{-it H_{\textrm{RWA}}}-e^{-it H})\Psi \right\| =0,
\end{equation}
uniformly for $t$ in compact sets.
\end{theorem}

This theorem proves~\eqref{eq:main}. In particular,~\eqref{approxim} shows that mathematically the rotating wave approximation is correct in the limit $\omega\to\infty$. How appropriate the approximation is in practice with finite parameters can be computed in~\eqref{upperbound}, which provides a concrete upper bound on the norm difference of the time evolution of an initial state under the actual time evolution and under the rotating wave approximation.

\begin{proof} First we prove~\eqref{upperbound}. Let $\Psi \in \C^2 \otimes \SS(\R)$ and $t \in \R$. First of all we observe that
\begin{equation}
	\| (I\otimes a) \Psi\|^2 = \langle (I\otimes a) \Psi|(I\otimes a)\Psi\rangle=\langle\Psi|(I \otimes a^{\dagger}a)\Psi\rangle=\langle\Psi| N \Psi\rangle = \| N^{1/2} \Psi\|^2,
\end{equation}
\begin{equation}
	\| (I\otimes a^{\dagger}) \Psi\|^2 =\langle\Psi| (N+1) \Psi\rangle = \| (N+1)^{1/2} \Psi\|^2.
\end{equation}
Moreover, the conservation law
\begin{equation}
	[H_2,\mathcal{N}]=0, \qquad 
	\mathcal{N}=P_{+} + N, \qquad P_{+}= \sigma_+\sigma_-\otimes I, \qquad  
\end{equation}
implies that
\begin{eqnarray}
	U_2(t)^{\dagger} N U_2(t) &=& U_2(t)^{\dagger}\mathcal{N} U_2(t) - U_2(t)^{\dagger}P_{+}U_2(t) = \mathcal{N}  - U_2(t)^{\dagger}P_{+}U_2(t)	\nonumber\\
	&=& N + \bigl(P_{+} - U_2(t)^{\dagger}P_{+}U_2(t)\bigr) \label{eqn:consN} \\
	& \leq&  N+1 \nonumber
\end{eqnarray}
on $\C^2 \otimes \SS(\R)$.

Start from the equality
	\begin{equation}
\bigl(U_{2}(t)-U_{1}(t)\bigr)\Psi =  -i S_{21}(t)U_{2}(t)\Psi+\int_{0}^{t}U_{1}(t)U_1(s)^{\dagger}\bigl(S_{21}(s)H_{2}-H_{1}(s)S_{21}(s)\bigr)U_{2}(s)\Psi\,ds \, .
\end{equation}
Let $u(t)=U_2(t)\Psi$. One gets
\begin{eqnarray}
	\|S_{21}(t)U_2(t)\Psi\|^2 
		&=&\langle S_{21}(t) u(t)  |  S_{21}(t) u(t) \rangle
	\nonumber\\
	&=& \left( \frac{\lambda\sin\left(t\omega\right)}{\omega}\right)^2
	\langle u(t)  | \bigl(\sigma_{+}\sigma_-\otimes a^{{\dagger}}a+ \sigma_{-}\sigma_+\otimes aa^{\dagger} \bigr
	)
	u(t) \rangle
	\nonumber\\
	&=& \left( \frac{\lambda\sin\left(t\omega\right)}{\omega}\right)^2
	\| \bigl(\sigma_{+}\sigma_{-}\otimes n^{1/2}+ \sigma_{-}\sigma_{+}\otimes (n+1)^{1/2} \bigr
	)
	u(t) \|^2
	\nonumber\\
	&\leq&  \frac{\lambda^2}{\omega^2} \| (N +1)^{1/2} u(t)\|^2
	\nonumber\\
	&\leq&  \frac{\lambda^2}{\omega^2}
	\| (N +2)^{1/2} \Psi\|^2 . \label{estact}
	\end{eqnarray} 
Moreover, let
\begin{equation}
	V(t)= H_2 - H_1(t) = -\lambda\left(e^{2it\omega}\sigma_{+}\otimes a^{{\dagger}}+e^{-2it\omega}\sigma_{-}\otimes a\right).
\end{equation}
Then
\begin{eqnarray}
	X(t) &=& S_{21}(t) H_2 - H_1(t)S_{21}(t) 
	\nonumber\\
	&=& [S_{21}(t), H_2] + V(t)S_{21}(t)
	\nonumber\\
	&=&
	-\frac{\lambda\sin\left(t\omega\right)}{\omega}\biggl[\Delta\left(-e^{it\omega}\sigma_{+}\otimes a^{{\dagger}}+e^{-it\omega}\sigma_{-}\otimes a\right)
	\nonumber\\
	& & \qquad\qquad\quad\;\; + \lambda \sigma_+\sigma_{-} \otimes \left(  e^{it\omega} a^{{\dagger}2} - e^{-it\omega} a^2 - e^{it\omega} a^{{\dagger}}a\right)
	\nonumber\\
	& & \qquad\qquad\quad\;\;+ \lambda \sigma_{-}\sigma_{+} \otimes \left( - e^{it\omega} a^{{\dagger}2} + e^{-it\omega} a^2 - e^{-it\omega} a a^{{\dagger}}\right)
	\biggr]. \label{eq:Xt}
\end{eqnarray}
We want to estimate $\|X(t)u(t)\|$. First we observe that
\begin{equation}
	\|\left(-e^{it\omega}\sigma_{+}\otimes a^{{\dagger}}+e^{-it\omega}\sigma_{-}\otimes a\right) u(t)\| 
		\leq \|(N+1)^{1/2} u(t)\| .
\end{equation}
Moreover for all $\psi \in \SS(\R)$:
\begin{eqnarray}
	\|\left(  e^{it\omega} a^{{\dagger}2} - e^{-it\omega} a^2 - e^{it\omega} a^{{\dagger}}a\right)\psi\| 
	&\leq& \|a^{{\dagger}2}\psi\|+ \|a^{2}\psi\| + \|a^{{\dagger}}a\psi\|
	\nonumber\\
	&\leq& 3\|((a^{\dagger}a+1)(a^{\dagger}a+2))^{1/2}\psi\|
	\end{eqnarray}
and
\begin{eqnarray}
	\|\left( - e^{it\omega} a^{{\dagger}2} + e^{-it\omega} a^2 - e^{-it\omega} a a^{{\dagger}}\right)\psi\| 
	&\leq& \|a^{{\dagger}2}\psi\|+ \|a^{2}\psi\| + \|aa^{{\dagger}}a\psi\|\nonumber
	\\
	&\leq& 3\|((a^{\dagger}a+1)(a^{\dagger}a+2))^{1/2}\psi\|.
\end{eqnarray}
Therefore,
\begin{eqnarray}
	\|X(t) u(t)\|&\leq& \frac{\lambda}{\omega} \Bigl[ |\Delta|\|(N+1)^{1/2} u(t)\| + 3 \lambda \|((N+1)(N+2))^{1/2}u(t)\|\Bigr]
	\nonumber\\
	&\leq& \frac{\lambda}{\omega} \Bigl[ |\Delta|\|(N+2)^{1/2} \Psi\| + 3 \lambda \|((N+2)(N+3))^{1/2}\Psi\|\Bigr]. \label{boundXt}
\end{eqnarray}
Taking things together we  get 
\begin{equation}
\left\Vert (U_{2}(t)-U_{1}(t))\Psi\right\Vert  \le 	 \frac{\lambda}{\omega} \biggl[
	\| (N +2)^{1/2} \Psi\| 
	+| t| \Bigl(| \Delta|\|(N+2)^{1/2} \Psi\| + 3 \lambda \|\bigl((N+2)(N+3)\bigr)^{1/2}\Psi\|\Bigr)	\biggr].
\end{equation}
Therefore
\begin{equation}
\lim_{\omega \to + \infty} \left\|  (e^{-it H_{\textrm{RWA}}}-e^{-it H})\Psi \right\| =
\lim_{\omega \to + \infty} \left\| (U_2(t)-U_1(t))\Psi \right\| =
0,
\end{equation}
for all $\Psi\in\C^2 \otimes \SS(\R)$, and since the latter is dense in $\C^2 \otimes L^2(\R)$,~\eqref{approxim} follows.

\end{proof}

\subsection{Lower bound}

\begin{theorem}
For all $\Psi \in \C^2 \otimes \SS(\R)$ and for all $0 \leq t \leq \pi/\omega$:
\begin{eqnarray}
\|(U_{2}(t) - U_{1}(t))\Psi \|  &\geq  & \frac{\lambda}{\omega} \sin(t\omega) \|(N-1)_{+}^{1/2}\Psi \|  \nonumber \\
                                              && - \frac{\lambda}{\omega^2}(1-\cos(t\omega))
\Bigl( |\Delta|\|(N+2)^{1/2} \Psi\| + 3 \lambda \|\bigl((N+2)(N+3)\bigr)^{1/2}\Psi\|\Bigr), \nonumber 
\end{eqnarray}
where $(N-1)_{+}$ denotes the positive part of the operator $N-1$. 
\end{theorem}

This theorem reveals the limitation of the rotating wave approximation in practice: the error grows with the photon number, so for larger systems the rotating wave approximation may no longer be justified.
\begin{proof}
Start from the equality
	\begin{equation}
\bigl(U_{2}(t)-U_{1}(t)\bigr)\Psi =  -i S_{21}(t)U_{2}(t)\Psi+\int_{0}^{t}U_{1}(t)U_1(s)^{\dagger}\bigl(S_{21}(s)H_{2}-H_{1}(s)S_{21}(s)\bigr)U_{2}(s)\Psi\,ds \, .
\end{equation}
We have that
\begin{eqnarray}
 \|(U_{2}(t) - U_{1}(t))\Psi\|  &\ge & \|S_{21}(t)\Psi \| - \int_{0}^{t} \|(S_{21}(s)H_{2}-H_{1}(s)S_{21}(s))U_{2}(s)\Psi\| \,ds \nonumber \\
\end{eqnarray}
We define $u(t)=U_2(t)\Psi$ and we have
\begin{eqnarray}
\left\Vert S_{21}(t)U_2(t)\Psi\right\Vert ^{2} & =& \left( \frac{\lambda\sin\left(t\omega\right)}{\omega}\right)^2
\langle u(t)  | \bigl(\sigma_{+}\sigma_-\otimes a^{{\dagger}}a+ \sigma_{-}\sigma_+\otimes aa^{{\dagger}} \bigr) u(t) \rangle\nonumber\\
&\ge& \left( \frac{\lambda\sin\left(t\omega\right)}{\omega}\right)^2 \|N^{1/2}u(t)\|^2,
\end{eqnarray}
moreover, by~(\ref{eqn:consN}), we have
\begin{equation}
U_2(t)^{\dagger} N U_2(t) \ge N-1.
\end{equation}
Hence, one can show that
\begin{equation}
\left\Vert S_{21}(t)U_2(t)\Psi\right\Vert ^{2} \ge \left( \frac{\lambda\sin\left(t\omega\right)}{\omega}\right)^2 \|(N-1)_{+}^{1/2}\Psi\|^2.
\end{equation}
Moreover, for $0\le t\le \pi/\omega$ we have 
\begin{eqnarray}
	&& \int_{0}^{t} \|(S_{21}(s)H_{2}-H_{1}(s)S_{21}(s))U_{2}(s)\Psi\| \,ds  \nonumber \\
	&&  \leq  \frac{\lambda}{\omega}\Bigl[ |\Delta|\|(N+2)^{1/2} \Psi\| + 3 \lambda \|((N+2)(N+3))^{1/2}\Psi\|\Bigr] \int_{0}^t \sin(\omega s)\, ds \nonumber \\
	&& =  \frac{\lambda(1-\cos(\omega t))}{\omega^2}\Bigl[ |\Delta|\|(N+2)^{1/2} \Psi\| + 3 \lambda \|((N+2)(N+3))^{1/2}\Psi\|\Bigr]
\end{eqnarray}
and hence, for $0\le t\le \pi/\omega$, we get
\begin{eqnarray}
 \|(U_{2}(t) - U_{1}(t))\Psi\|  &\ge & \|S_{21}(t)\Psi \| - \int_{0}^{t} \|(S_{21}(s)H_{2}-H_{1}(s)S_{21}(s))U_{2}(s)\Psi\| \,ds \nonumber \\
                                           &\ge & \frac{\lambda}{\omega} \sin(t\omega)\|(N-1)_{+}^{1/2}\Psi\|   \\
                                           && - \frac{\lambda}{\omega^2}(1-\cos(t\omega))
\Bigl( |\Delta|\|(N+2)^{1/2} \Psi\| + 3 \lambda \|\bigl((N+2)(N+3)\bigr)^{1/2}\Psi\|\Bigr).\nonumber
\end{eqnarray}
\end{proof}

\subsection{Applying the bounds to Fock states}
To understand the scaling better, let us apply the above bounds to (normalised) Fock states $\Phi_{j,n}= e_j \otimes \varphi_n \in \C^2 \otimes \SS(\R)$, with $j \in \{1,2\}$ and $n\in \N$, $n>0$.  For simplicity, we consider the case  $\Delta=0$, but the argument is easily generalised. From the upper bound~\eqref{upperbound} we obtain
\begin{equation}
\left\Vert (U_{2}(t)-U_{1}(t))\Phi_{j,n}\right\Vert  \le 	 \frac{\lambda (n +2)^{1/2}}{\omega} \biggl[
	 1 
	+ 3|t| \lambda  (n+3)^{1/2}	\biggr].
\end{equation}
For the lower bound we have, for $0\le t\le \pi/\omega$,
\begin{align}\label{lowerbound}
\|(U_{2}(t) -& U_{1}(t))\Phi_{j,n}\|  \ge \frac{\lambda}{\omega} \sin(t\omega)(n-1)^{1/2}  - \frac{3\lambda ^2}{\omega^2}(1-\cos(t\omega))
 \bigl((n+2)(n+3)\bigr)^{1/2}.
\end{align}
We compute that this as a function of $t$ has a maximum at
\begin{equation}\label{time}
t_*=\frac{\cos ^{-1}\left(\frac{3 \lambda }{\sqrt{9 \lambda ^2+\frac{(n-1) \omega ^2}{(n+2) (n+3)}}}\right)}{\omega }
\end{equation}
At this time, the right hand side of Eq.~\eqref{lowerbound} evaluates as
\begin{equation}\label{loweratto}
	\frac{g \left(9 g^2 (n+2) (n+3)-3 g \sqrt{(n+2) (n+3) \left(9 g^2 (n+2) (n+3)+n-1\right)}+n-1\right)}{\sqrt{9 g^2 (n+2) (n+3)+n-1}}.
\end{equation}
Here, we have set $g\equiv\frac{\lambda}{\omega}.$

We can expand this in order of $n^{-1}$ to obtain a slightly simpler exact lower bound (valid for $n>0$)
\begin{align}\label{lowerbound1}
\sup_{t\in [0,\frac{\pi }{\omega}]}\|(U_{2}(t) -& U_{1}(t))\Phi_{j,n}\|  \ge \frac{1}{6}-\frac{1}{216 g^2 n}-\frac{7}{12 n} .
\end{align}
Focussing on the same time interval also for the upper bound and linearising it, we can conclude that\begin{align}\label{lowerbound4}
5g\sqrt{n +3} \ge \sup_{t\in [0,\frac{\pi }{\omega}]}\|(U_{2}(t) -& U_{1}(t))\Phi_{j,n}\|  \ge \frac{1}{6}-\frac{1}{216 g^2 n}-\frac{7}{12 n},
\end{align}
proving~\eqref{lowerbound2}. This bound is not necessarily a sharp bound but it shows nicely that the short-time error becomes small for small $g$ and large for high photon number $n$, hence it provides us with a quantitative condition on $g$ in order to reduce the error below a certain bound for given photon number. For high photon numbers $n\to\infty$, there is a time such that the difference becomes greater than $\frac16$, i.e.,
\begin{equation}
\|e^{-itH}-e^{-itH_{\textrm{RWA}}}\|\ge \frac{1}{6}
\end{equation}
by taking the supremum over all states in~\eqref{lowerbound1}, which means that the rotating wave approximation does not work for arbitrarily high photon numbers; this proves~\eqref{lb}.


\begin{thebibliography}{10}

\bibitem{rabi_new_1938}
I.~I. Rabi, J.~R. Zacharias, S.~Millman, and P.~Kusch, A {New} {Method} of
  {Measuring} {Nuclear} {Magnetic} {Moment}. \href{https://doi.org/10.1103/PhysRev.53.318}{Physical Review \textbf{53}, 318 (1938).}

\bibitem{rabi_space_1937}
I.~I. Rabi, Space {Quantization} in a {Gyrating} {Magnetic} {Field}. . \href{https://doi.org/10.1103/PhysRev.51.652}{Physical Review \textbf{51}, 652 (1937).}

\bibitem{pop}
Google Scholar reports close to a million hits.

\bibitem{bloch_magnetic_1940}
F.~Bloch and A.~Siegert, Magnetic {Resonance} for {Nonrotating} {Fields}.
  \href{https://doi.org/10.1103/PhysRev.57.522}{Physical Review \textbf{57}, 522 (1940).}

\bibitem{Shirley1965}
J.~H. Shirley, Solution of the schr\"{o}dinger equation with a hamiltonian
  periodic in time. \href{https://doi.org/10.1103/PhysRev.138.B979}{Physical Review \textbf{138}, B979 (1965).}

\bibitem{haeberlen_coherent_1968}
U.~Haeberlen and J.~S. Waugh, Coherent {Averaging} {Effects} in {Magnetic}
  {Resonance}. \href{https://doi.org/10.1103/PhysRev.175.453}{Physical Review \textbf{175}, 453 (1968).}

\bibitem{agarwal_rotating-wave_1973}
G.~S. Agarwal, Rotating-{Wave} {Approximation} and {Spontaneous}
  {Emission}. \href{https://doi.org/10.1103/PhysRevA.7.1195}{Physical Review A \textbf{7}, 1195 (1973).}

\bibitem{burgarth_one_2022}
D.~Burgarth, P.~Facchi, G.~Gramegna, and K.~Yuasa, One bound to rule them
  all: from {Adiabatic} to {Zeno}.  \href{https://doi.org/10.22331/q-2022-06-14-737}{Quantum \textbf{6}, 737 (2022).}
  
\bibitem{jaynes_comparison_1963}
E.~Jaynes and F.~Cummings, Comparison of quantum and semiclassical radiation
  theories with application to the beam maser. \href{https://doi.org/10.1109/PROC.1963.1664}{Proceedings of the IEEE \textbf{51}, 89 (1963).}
  
  
\bibitem{larson_jaynes-cummings_2021}
J.~Larson and T.~Mavrogordatos, The {Jaynes}-{Cummings} model and its
  descendants modern research directions. IoP Publishing (2021).

  
\bibitem{forn-diaz_observation_2010}
P.~Forn-D\'{i}az, J.~Lisenfeld, D.~Marcos, J.~J. Garc\'{i}a-Ripoll, E.~Solano, C.~J.
  P.~M. Harmans, and J.~E. Mooij, Observation of the {Bloch}-{Siegert}
  {Shift} in a {Qubit}-{Oscillator} {System} in the {Ultrastrong} {Coupling}
  {Regime}. \href{https://doi.org/10.1103/PhysRevLett.105.237001}{Physical Review Letters \textbf{105}, 237001 (2010).}

\bibitem{li_vacuum_2018}
X.~Li, M.~Bamba, Q.~Zhang, S.~Fallahi, G.~C. Gardner, W.~Gao, M.~Lou,
  K.~Yoshioka, M.~J. Manfra, and J.~Kono, Vacuum {Bloch}-{Siegert} shift in
  {Landau} polaritons with ultra-high cooperativity. \href{https://doi.org/10.1038/s41566-018-0153-0}{Nature Photonics \textbf{12}, 324 (2018).}
  
  \bibitem{Nori2019}
  A.~Frisk, A.~Miranowicz, S.~De Liberato, S.~Savasta, and F.~Nori, Ultrastrong coupling between light and matter. \href{https://doi.org/10.1038/s42254-018-0006-2}{Nature Review Physics \textbf{1}, 19 (2019).}
  
\bibitem{walls_higher_1972}
D.~Walls, Higher order effects in the single atom field mode interaction. \href{https://doi.org/10.1016/0375-9601(72)90867-5}{Physics Letters A \textbf{42}, 217 (1972).}


\bibitem{wang_photon-dressed_2020}
S.-P. Wang, G.-Q. Zhang, Y.~Wang, Z.~Chen, T.~Li, J.~S. Tsai, S.-Y. Zhu, and
  J.~Q. You, Photon-{Dressed} {Bloch}-{Siegert} {Shift} in an {Ultrastrongly}
  {Coupled} {Circuit} {Quantum} {Electrodynamical} {System}. \href{https://doi.org/10.1103/PhysRevApplied.13.054063}{Physical
  Review Applied \textbf{13} ,054063 (2020).}

\bibitem{puri_mathematical_2011}
R.~R. Puri, Mathematical Methods of Quantum Optics. Springer (2011).

\bibitem{remark}
The reason we focus on short times here is a purely technical requirement from the proof (see Appendix). Indeed, numerics shows that the errors are even larger for generic later times.


\bibitem{Gottesman2022}
D.~Gottesman, Opportunities and Challenges in Fault-Tolerant Quantum Computation. \href{https://doi.org/10.48550/arXiv.2210.15844}{arXiv:2210.15844 (2022).}

\bibitem{Grimsmo2021}
A.~Grimsmo and S.~Puri, Quantum Error Correction with the Gottesman-Kitaev-Preskill Code. \href{https://doi.org/10.1103/PRXQuantum.2.020101}{PRX Quantum, \textbf{2}, 020101 (2021).}

\bibitem{vlastakis_deterministically_2013}
B.~Vlastakis, G.~Kirchmair, Z.~Leghtas, S.~E. Nigg, L.~Frunzio, S.~M. Girvin,
  M.~Mirrahimi, M.~H. Devoret, and R.~J. Schoelkopf, Deterministically
  {Encoding} {Quantum} {Information} {Using} 100-{Photon} {Schr\"{o}dinger}
  {Cat} {States}. \href{https://doi.org/10.1126/science.1243289}{Science \textbf{342}, 607 (2013).}

\bibitem{dda} D. D'Alessandro, \textit{Introduction to Quantum Control and Dynamics}, CRC Press (2020).

\bibitem{chambrion} T. Chambrion, Periodic excitations of bilinear quantum systems.  \href{https://doi.org/10.1016/j.automatica.2012.03.031}{Automatica \textbf{48}, 2040 (2012).}
\bibitem{augier} N. Augier, U. Boscain, and M. Sigalotti,  Effective adiabatic control of a decoupled Hamiltonian obtained by rotating wave approximation. \href{https://doi.org/10.1016/j.automatica.2021.110034}{Automatica \textbf{136}, 110034 (2022).}

\bibitem{robin} R. Robin, N. Augier, U. Boscain, and M.  Sigalotti. Ensemble qubit controllability with a single control via adiabatic and rotating wave approximations. \href{https://doi.org/10.1016/j.jde.2022.02.042}{Journal of Differential Equations \textbf{318}, 414 (2022).}

\bibitem{RS1} M. Reed and B. Simon. \textit{Methods of mathematical physics 1. Functional analysis}. Academic Press (1980).
\bibitem{RS2} M. Reed and B. Simon.
\textit{Methods of mathematical physics 2. Fourier analysis, self-adjointness}. Academic Press (1975).


\end{thebibliography}
\end{document}